\documentclass[12pt,fleqn]{article}

\usepackage{amsmath,amssymb,amsthm}
\usepackage[margin=1in]{geometry}
\usepackage{xcolor}
\usepackage{url}
\usepackage{float}
\usepackage{array}
\usepackage{comment}
\usepackage{tikz}
\usetikzlibrary{arrows.meta}
\usepackage{subcaption}
\usepackage{listings}
\usepackage[hypertexnames=false,bookmarksnumbered=true,final]{hyperref}
\usepackage{algorithm,algpseudocode}
\usepackage[capitalize,sort]{cleveref}

\def\colorschemesepia{sepia}
\def\colorschemedark{dark}
\def\colorschemelight{light}

\ifx\colorscheme\undefined
\let\colorscheme\colorschemelight
\fi

\ifx\colorscheme\colorschemelight
\colorlet{textColor}{black}
\colorlet{bgColor}{white}
\fi

\ifx\colorscheme\colorschemesepia
\definecolor{textColor}{HTML}{433423}
\definecolor{bgColor}{HTML}{fbf0da}
\fi

\ifx\colorscheme\colorschemedark
\definecolor{textColor}{HTML}{bdc1c6}
\definecolor{bgColor}{HTML}{202124}
\definecolor{textBlue}{HTML}{8ab4f8}
\definecolor{textRed}{HTML}{f9968b}
\definecolor{textGreen}{HTML}{81e681}
\definecolor{textPurple}{HTML}{c58af9}
\else
\colorlet{textBlue}{blue!50!black}
\colorlet{textRed}{red!50!black}
\colorlet{textGreen}{green!50!black}
\definecolor{textPurple}{HTML}{681da8}
\fi

\ifx\colorscheme\colorschemelight\else
\pagecolor{bgColor}
\color{textColor}
\fi

\hypersetup{colorlinks,linkcolor=textRed,citecolor=textRed,urlcolor=textBlue}

\let\eps\varepsilon
\newcommand*{\defeq}{:=}
\newcommand*{\Th}{^{\textrm{th}}}

\DeclareMathOperator*{\E}{\mathbb{E}}

\DeclareMathOperator*{\argmin}{argmin}
\DeclareMathOperator*{\argmax}{argmax}

\newcommand*{\wLoG}{without loss of generality}

\DeclareMathOperator*{\argsup}{argsup}
\DeclareMathOperator{\opt}{opt}
\DeclareMathOperator{\LPT}{LPT}
\newcommand*{\Null}{\mathtt{null}}

\newcommand*{\Ical}{\mathcal{I}}
\newcommand*{\xhat}{\widehat{x}}
\newcommand*{\zhat}{\widehat{z}}

\newtheorem{theorem}{Theorem}

\newtheorem{lemma}[theorem]{Lemma}

\makeatletter
\g@addto@macro{\UrlBreaks}{%
\do\/%
\do\a\do\b\do\c\do\d\do\e\do\f\do\g\do\h\do\i\do\j\do\k\do\l\do\m%
\do\n\do\o\do\p\do\q\do\r\do\s\do\t\do\u\do\v\do\w\do\x\do\y\do\z%
\do\A\do\B\do\C\do\D\do\E\do\F\do\G\do\H\do\I\do\J\do\K\do\L\do\M%
\do\N\do\O\do\P\do\Q\do\R\do\S\do\T\do\U\do\V\do\W\do\X\do\Y\do\Z%
\do\0\do\1\do\2\do\3\do\4\do\5\do\6\do\7\do\8\do\9%
}
\makeatother

\allowdisplaybreaks
\newenvironment*{tightenum}{\begin{enumerate}\setlength{\itemsep}{0.0em}}{\end{enumerate}}%

\tikzset{
cVertex/.style={thick,minimum size=7mm,inner sep=0},
cGrayVertex/.style={cVertex,draw=textColor,fill={textColor!7!bgColor}},
cBlueVertex/.style={cVertex,draw={blue!20!textColor},fill={blue!10!bgColor},text=textBlue},
cVar/.style={cBlueVertex,shape=rectangle},
cConst/.style={cGrayVertex,shape=circle},
cOp/.style={cGrayVertex,shape=circle},
cOut/.style={double=bgColor},
cEdge/.style={->,>={Stealth},semithick},
}

\definecolor{codeBlue}{HTML}{268bd2}
\definecolor{codeCyan}{HTML}{2aa198}
\definecolor{codeGreen}{HTML}{859900}
\ifx\colorscheme\colorschemedark
\definecolor{codeBg}{HTML}{002b36}
\definecolor{codeColor}{HTML}{839496}
\definecolor{codeSecColor}{HTML}{586e75}
\else
\definecolor{codeBg}{HTML}{fdf6e3}
\definecolor{codeColor}{HTML}{657b83}
\definecolor{codeSecColor}{HTML}{93a1a1}
\fi

\lstdefinestyle{myStyle}{
backgroundcolor=\color{codeBg!50!bgColor},
commentstyle=\color{codeSecColor},
keywordstyle=\color{codeGreen},
numberstyle=\color{codeCyen},
stringstyle=\color{codeCyan},
basicstyle={\small\ttfamily\color{codeColor!80!textColor}},
showspaces=false,
showstringspaces=false,
}

\title{Automating the Search for Small Hard Examples to Approximation Algorithms}
\author{Eklavya Sharma\\
Industrial \& Enterprise Systems Engineering\\
University of Illinois at Urbana-Champaign, USA\\
\texttt{eklavya2@illinois.edu}}
\date{\empty}

\begin{document}

\maketitle
\setlength{\parskip}{0.5em}

\begin{abstract}
Given an approximation algorithm $A$,
we want to find the input with the worst approximation ratio,
i.e., the input for which $A$'s output's objective value is the worst possible
compared to the optimal solution's objective value.
Such \emph{hard examples} shed light on the approximation algorithm's weaknesses,
and could help us design better approximation algorithms.
When the inputs are discrete (e.g., unweighted graphs),
one can find hard examples for small input sizes using brute-force enumeration.
However, it's not obvious how to do this when the input space is continuous,
as in makespan minimization or bin packing.

We develop a technique for finding small hard examples
for a large class of approximation algorithms.
Our algorithm works by constructing a decision tree representation of the approximation algorithm
and then running a linear program for each leaf node of the decision tree.
We implement our technique in Python, and demonstrate it on the
longest-processing-time (LPT) heuristic for makespan minimization.

\end{abstract}

\section{Introduction}
\label{sec:intro}

Optimization problems, and algorithms to solve them,
have been of immense importance in computer science and operations research.
Discrete optimization problems like shortest path in graphs, vertex cover,
bin packing, and job scheduling, and continuous optimization problems
like linear programming and max $s$-$t$ flow, have been intensely studied.

Often, optimization problems cannot be solved optimally
due to computational hardness (e.g., NP-hardness),
incomplete or noisy input (e.g., in online or robust optimization regimes),
or external constraints (e.g., restriction to structured solutions or strategyproof mechanisms).
Hence, approximation algorithms are widely used.

A central question in approximation algorithm design is finding the
\emph{approximation ratio} of an algorithm.
The approximation ratio is defined as the worst-case ratio of
the algorithm's output's objective value and the optimal solution's objective value.
Formally, let $A(I)$ be algorithm $A$'s output on input $I$,
let $\Ical$ be the set of all inputs,
and let $\opt(I)$ be the optimal solution to input $I$.
For minimization problems, let $c(I, z)$ be the cost of solution $z$ for input $I$.
Then the approximation ratio of $A$ is defined as
\[ \alpha(A) \defeq \sup_{I \in \Ical} \frac{c(I, A(I))}{c(I, \opt(I))}. \]
Hence, the approximation ratio is at least 1 for minimization problems.
For maximization problems, let $s(I, z)$ be the score of solution $z$ for input $I$.
Then $A$'s approximation ratio is defined as
\[ \alpha(A) \defeq \inf_{I \in \Ical} \frac{s(I, A(I))}{s(I, \opt(I))}. \]
Hence, the approximation ratio is at most 1 for maximization problems.

For some approximation algorithms, determining their exact approximation ratio is difficult and
requires deep insight into the problem and the algorithm
(though obtaining loose lower and upper bounds may be easy).
One way to build intuition is to run the algorithm on many small inputs
and see which inputs the algorithm performs poorly on.
Formally, for minimization problems, let $\Ical_n$ be the set of all inputs of size $n$.
Then the \emph{approximation ratio for size $n$} is defined as
\[ \alpha_n(A) \defeq \sup_{I \in \Ical_n} \frac{c(I, A(I))}{c(I, \opt(I))}. \]
For problems with discrete inputs (e.g., min-cardinality vertex cover),
$\Ical_n$ is a finite set, so for small enough values of $n$,
one can find $\alpha_n(A)$ and the corresponding maximizer in $\Ical_n$ by brute force.
These maximizers, called \emph{hard inputs}, shed light on the algorithm's weaknesses,
and help us obtain tight bounds on the approximation ratio of $A$,
and may even help us design better approximation algorithms.

This approach fails for problems with continuous inputs.
Let us take the makespan minimization problem as an example.
In this problem, we are given $n$ jobs and $m$ identical machines,
(assume $m$ is a fixed parameter, not part of the input),
and our goal is to assign each job to a machine such that
the maximum load among the machines is minimized.
The input to this problem is a vector $x \in \mathbb{R}^n$,
where $x_j$ is the size of job $j$. Since the set of inputs is continuous,
one cannot find $\alpha_n(A)$ by brute force for any algorithm $A$.

We show that for problems with continuous inputs and discrete outputs
(e.g., bin packing, makespan minimization),
for a large family of approximation algorithms,
one can obtain the approximation ratio for size $n$
(and the corresponding hard input) in finite time.
We do this using a mix of brute force and linear programming.
What makes our approach special is that it is fully automatic.
Specifically, we give a computer program that takes input $(n, A)$,
where $n \in \mathbb{N}$ and $A$ is an approximation algorithm
written in an appropriate programming language (e.g., Python).
The program outputs $\alpha_n(A)$ and the corresponding hard input.

Our program works for every algorithm $A$ that can be represented as
a decision tree for any fixed input size $n \in \mathbb{N}$.
Many well-known approximation algorithms can be represented this way.
Our program works by computing $A$'s decision tree
and running a linear program for each leaf of $A$.

Using our program, we could recover known hard examples
for the longest-processing-time (LPT) heuristic \cite{graham1969bounds}
for the makespan minimization problem with up to 4 machines.

We used a rudimentary form of this paper's techniques to find a hard example for
Garg and Taki's algorithm \cite{garg2021improved} for the fair division of goods.
We presented our hard example in Section 5 of \cite{akrami2023simplification},
though we later found out that this example was already known for the LPT heuristic
for the same problem \cite{deuermeyer1982scheduling,babaioff2021fair}.

In \cref{sec:prelims}, we formally define optimization problems and decision trees.
In \cref{sec:che}, we show how to find hard inputs for an approximation algorithm
using the algorithm's decision tree.
In \cref{sec:cdt}, we show how to construct an algorithm's decision tree.
In \cref{sec:implement-lpt}, we describe an implementation of our ideas for
the LPT heuristic for makespan minimization.
We end with concluding remarks in \cref{sec:conclusion}.

\section{Preliminaries}
\label{sec:prelims}

For any non-negative integer $n$, let $[n] \defeq \{1, 2, \ldots, n\}$.
Let $\mathbb{N} \defeq \{1, 2, 3, \ldots\}$.

\paragraph{Division by 0:} For any $a \in \mathbb{R} \setminus \{0\}$,
define $a/0 \defeq \infty$ if $a > 0$ and $a/0 \defeq -\infty$ if $a < 0$.

\subsection{Suprema and Infima}
\label{sec:prelims:sup-inf}

We define the notation and conventions for supremum and maximum.
The notation and conventions for infimum and minimum can be defined analogously.

Let $Y \subseteq \mathbb{R} \cup \{-\infty, \infty\}$.
If $Y = \emptyset$ or $Y = \{-\infty\}$, let $\sup(Y) \defeq -\infty$.
If $\infty \in Y$ or if $Y$ has no upper bound in $\mathbb{R}$, let $\sup(Y) \defeq \infty$.
Otherwise, $\sup(Y)$ is the least upper bound on $Y$.
If $\sup(Y) \in Y$, then $\max(Y) \defeq \sup(Y)$,
otherwise $\max(Y)$ is said to be undefined or non-existent.

Let $X$ be any set and $f: X \to \mathbb{R} \cup \{-\infty, \infty\}$ be any function.
Define $Y \defeq \{f(x): x \in X\}$.
Define $\sup_{x \in X} f(x) \defeq \sup(Y)$
and $\max_{x \in X} f(x) \defeq \max(Y)$. Define
\[ \argmax_{x \in X} f(x) \defeq \left\{x \in X: f(x) = \max(Y) \right\}. \]
If $\max(Y)$ is not defined,
we can use $\argsup$ as an alternative to $\argmax$.
For $X \neq \emptyset$, let $\argsup_{x \in X} f(x)$ be the set of all infinite sequences
$(x_1, x_2, \ldots)$ where $(f(x_1), f(x_2), \ldots)$ is non-decreasing, and
\begin{tightenum}
\item If $\sup(Y) = \infty$, then $\forall M \in \mathbb{R}$, $\exists i \in \mathbb{N}$
    such that $f(x_i) \ge M$.
\item If $\sup(Y) = -\infty$, then $f(x_i) = -\infty$ for all $i \in [n]$.
\item If $\sup(Y) \in \mathbb{R}$, then $\forall \eps > 0$, $\exists i \in \mathbb{N}$
    such that $\sup(Y) - \eps \le f(x_i) \le \sup(Y)$.
\end{tightenum}
It is easy to see that $\argsup_{x \in X} f(x)$ is non-empty (if $X \neq \emptyset$).

\subsection{Halfspaces and Polyhedra}
\label{sec:prelims:polyhedra}

For $a \in \mathbb{R}^n$ and $b \in \mathbb{R}$,
\begin{tightenum}
\item A set of the form $\{x \in \mathbb{R}^n: a^Tx \ge b\}$ is called a \emph{closed halfspace}.
\item A set of the form $\{x \in \mathbb{R}^n: a^Tx > b\}$ is called an \emph{open halfspace}.
\item A set of the form $\{x \in \mathbb{R}^n: a^Tx = b\}$ is called a \emph{hyperplane}.
\end{tightenum}
The intersection of closed halfspaces and hyperplanes is called a \emph{closed polyhedron}.
The intersection of closed and open halfspaces and hyperplanes is called a \emph{polyhedron}.

In the \emph{linear programming} (LP) problem,
we must maximize a linear objective under a (closed) polyhedral constraint.
This problem is of immense importance, and has been studied intensely.
\Cref{thm:lin-prog} describes the computational and structural guarantees for linear programming.

\begin{theorem}[Linear programming, \cite{bazaraa2010linear}]
\label{thm:lin-prog}
Given a (closed) polyhedron $P \subseteq \mathbb{R}^n$ and a vector $c \in \mathbb{R}^n$,
there is an efficient algorithm that outputs a pair $(x, d)$, where
$x \in P \cup \{\Null\}$ and $d \in \mathbb{R}^n$ such that
\begin{tightenum}
\item (infeasible) If $x = \Null$, then $P = \emptyset$.
\item (unbounded) If $x \neq \Null$ and $d \neq 0$, then $c^Td > 0$
    and $x + \alpha d \in P$ for all $\alpha \in \mathbb{R}_{\ge 0}$.
\item (optimal) If $x \neq \Null$ and $d = 0$, then $c^Tx \ge c^Ty$ for all $y \in P$.
\end{tightenum}
\end{theorem}

Although \cref{thm:lin-prog} only works for closed polyhedra, it can essentially
be made to work for non-closed polyhedra too. See \cref{sec:open-lp} for details.

\subsection{Optimization Problems and Algorithms}
\label{sec:optimization}

We describe minimization problems for a fixed input size $n \in \mathbb{N}$.
A minimization problem is given by a tuple $(X, Z, c)$.
Here $X \subseteq \mathbb{R}^n$ is the \emph{input space}, $Z$ is the \emph{output space},
and $c: X \times Z \to \mathbb{R}_{\ge 0} \cup \{\infty\}$ is the \emph{cost function},
i.e., $c(x, z)$ is the cost of output $z$ for input $x$.
Given $x \in X$, our goal is to output $z \in Z$ such that $c(x, z)$ is minimized.
We assume that $Z$ is finite, and $\forall x \in X$,
$\exists z \in Z$ such that $c(x, z) \in \mathbb{R}_{\ge 0}$.

An algorithm for an optimization problem is a function from $X$ to $Z$.
$A$'s approximation ratio $\alpha_n(A)$ and hard example set $H_n(A)$ are defined as
\begin{align*}
\alpha_n(A) &\defeq \sup_{x \in X} \frac{c(x, A(x))}{\displaystyle \inf_{z \in Z} c(x, z)}.
& H_n(A) &\defeq \argsup_{x \in X} \frac{c(x, A(x))}{\displaystyle \inf_{z \in Z} c(x, z)}.
\end{align*}

A minimization problem is called \emph{linear} if
$c(x, z)$ is piecewise-affine in $x$ for every $z \in Z$,
where each piece is a polyhedron.
Specifically, for all $z \in Z$, we have $X = \bigcup_{j=1}^{k_z} X^{(z)}_j$
such that $\forall j \in [k_z]$, $X_j$ is a polyhedron,
and for some $a \in \mathbb{R}^n$ and $b \in \mathbb{R} \cup \{\infty\}$,
we have $c(x, z) = a^Tx + b$ for all $x \in X^{(z)}_j$.

A maximization problem is given by a tuple $(X, Z, s)$,
where $X \subseteq \mathbb{R}^n$ is the \emph{input space}, $Z$ is the \emph{output space},
and $s: X \times Z \to \mathbb{R}_{\ge 0}$ is the \emph{score function},
i.e., $s(x, z)$ is the score of output $z$ for input $x$. $Z$ is finite.
Given $x \in X$, our goal is to output $z \in Z$ such that $s(x, z)$ is maximized.
We can define approximation ratios for algorithms for maximization problems analogously.

\subsection{Makespan Minimization}
\label{sec:makespan}

We will use the makespan minimization problem for illustrative purposes throughout this paper.
In this problem, we are given $n$ jobs and $m$ identical machines,
($m$ is a fixed parameter, not part of the input),
and our goal is to assign each job to a machine such that
the maximum load among the machines is minimized.

The input is a vector $x \in \mathbb{R}^n$, where $x_j$ is the size of job $j$.
Assume \wLoG{} that jobs are sorted in decreasing order of size, so
$X \defeq \{x \in \mathbb{R}^n: x_1 \ge \ldots \ge x_n \ge 0\}$.
We need to output an \emph{assignment}, i.e., a vector $(z_1, \ldots, z_n)$
where $z_j$ is the machine that job $j$ is assigned to. Hence, $Z \defeq [m]^n$.

The \emph{load} on machine $i$ is given by
\[ \ell_i(x, z) \defeq \sum_{j=1}^n x_j\left\{\begin{array}{ll}
1 & \quad\textrm{if } z_j = i
\\ 0 & \quad\textrm{if } z_j \neq i
\end{array}\right\}. \]
The cost of assignment $z$ is $c(x, z) \defeq \max_{i=1}^m \ell_i(x, z)$
(i.e., the maximum load among the machines).

The LPT (longest processing time) algorithm \cite{graham1969bounds}
for makespan minimization first sorts the jobs in non-increasing order of running times,
and then assigns the jobs sequentially and greedily,
i.e., in the $j\Th$ step, it assigns job $j$ to the machine with the least load so far.

\begin{algorithm}[htb]
\caption{$\LPT(x, m)$: Longest Processing Time algorithm.
\\ $m \in \mathbb{N}$ and $x \in \mathbb{R}^n$ such that $x_1 \ge \ldots \ge x_n \ge 0$.}
\label{algo:lpt}
\begin{algorithmic}[1]
\State Initialize \texttt{loads} to be a sequence of $m$ zeros.
\State Initialize $z$ to be an empty sequence.
\For{$j \in [n]$}
    \State Let $i \in \argmin_{i=1}^m \texttt{loads}_i$.
    \State Append $i$ to $z$.
    \State $\texttt{loads}_i \texttt{ += } x_j$.
\EndFor
\State \Return $z$
\end{algorithmic}
\end{algorithm}

$c(x, z)$ is piecewise-linear. For any $z \in Z$, define
$X^{(z)}_i \defeq \{x \in X: \ell_i(x, z) \ge \ell_j(x, z) \forall j \in [m] \setminus \{i\}\}$.
Then $c(x, z) = \ell_i(x, z)$ for all $x \in X^{(z)}_i$.

\subsection{Decision Tree}

A decision tree $D$ for a minimization problem $(X, Z, c)$ is a binary tree where
each internal node has exactly two children: a \emph{false child} and a \emph{true child}.
Each leaf node $v$ is labeled with an output $z_v \in Z$,
and each internal node $v$ is labeled with a function $f_v: \mathbb{R}^n \to \{0, 1\}$.
$D$ is called \emph{linear} if each $f_v$ is a (closed or open) halfspace.

For any $x \in \mathbb{R}^n$ and decision tree $D$, we define $D(x)$ to be
a leaf $\ell$ obtained by tracing a path from the root to $\ell$,
where for each internal node $v$, we move to the true child if $f_v(x) = 1$
and we move to the false child if $f_v(x) = 0$.
Hence, we can interpret the decision tree $D$ as a function from $\mathbb{R}^n$ to $D$'s leaves.

We say that algorithm $A$ can be represented as a decision tree $D$
if for every input $x \in X$, we get $A(x) = z_{D(x)}$.

For $m=2$ and $n=5$, the decision tree of the LPT algorithm is given in \cref{fig:lpt-dtree}.

\begin{figure}[htb]
\centering
\begin{tikzpicture}[
level 1/.style={sibling distance=6cm,level distance=2cm},
level 2/.style={sibling distance=3cm,level distance=2cm},
vertex/.style={draw, minimum size=0.7cm, rectangle, inner sep=5pt},
leaf/.style={vertex},
]
\node[vertex] (v) {$x_1 \le x_2 + x_3$}
child {node[vertex] (v0) {$x_1 \le x_2 + x_3 + x_4$}
    child {node[leaf] (v00) {$(1, 2, 2, 2, 2)$}}
    child {node[leaf] (v01) {$(1, 2, 2, 2, 1)$}}
}
child {node[vertex] (v1) {$x_1 + x_4 \le x_2 + x_3$}
    child {node[leaf] (v10) {$(1, 2, 2, 1, 2)$}}
    child {node[leaf] (v11) {$(1, 2, 2, 1, 1)$}}
};
\path (v) -- node[right] {true} (v1);
\path (v) -- node[left] {false} (v0);
\path (v1) -- node[right] {true} (v11);
\path (v1) -- node[left] {false} (v10);
\path (v0) -- node[right] {true} (v01);
\path (v0) -- node[left] {false} (v00);
\end{tikzpicture}

\caption{Decision tree for the LPT algorithm when $m=2$ and $n=5$ (with appropriate tie-breaking rules).}
\label{fig:lpt-dtree}
\end{figure}
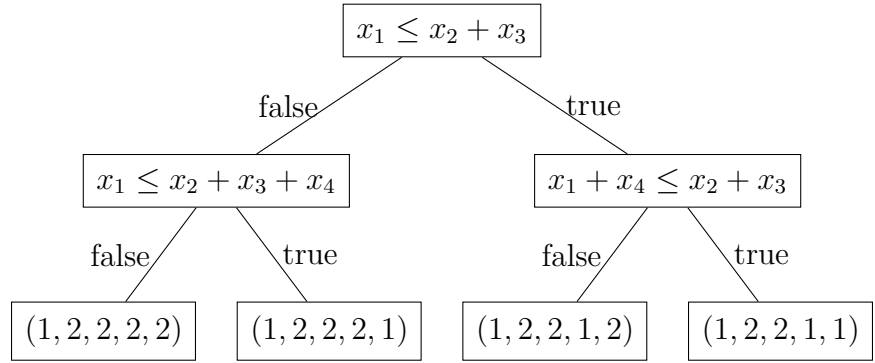

\section{Computing Hard Examples}
\label{sec:che}

Consider a linear minimization problem $(X, Z, c)$,
where $X \subseteq \mathbb{R}^n$ and $Z$ is finite.
Suppose we are given a linear decision tree $D$ for approximation algorithm $A$.
We show how to use $D$ to find the approximation ratio $\alpha_n(A)$
and hard examples $H_n(A)$.

Let $L$ be the set of all leaves of $D$.
For any leaf $\ell \in L$, let $X_{\ell} \defeq \{x \in X: D(x) = \ell\}$.
Then $X_{\ell}$ is the intersection of $X$ with halfspaces from
all the ancestors of $\ell$. Hence, $X_{\ell}$ is a polyhedron.
Furthermore, $X = \bigcup_{\ell \in L} X_{\ell}$.
\begin{align*}
\alpha_n(A) &\defeq \sup_{x \in X}\; \frac{c(x, A(x))}{\displaystyle \inf_{z^* \in Z} c(x, z^*)}
    = \max_{z^* \in Z}\; \sup_{x \in X}\; \frac{c(x, A(x))}{c(x, z^*)}
\\ &= \max_{z^* \in Z}\; \max_{\ell \in L}\; \sup_{x \in X_{\ell}}\;
    \frac{c(x, z_{\ell})}{c(x, z^*)}.
\\ &= \max_{z^* \in Z}\; \max_{\ell \in L}\;
    \max_{j_{\ell} \in [k_{z_{\ell}}]}\; \max_{j^* \in [k_{z^*}]}\;
    \sup_{x \in X_{\ell} \cap X^{(z_{\ell})}_{j_{\ell}} \cap X^{(z^*)}_{j^*}}\;
    \frac{c(x, z_{\ell})}{c(x, z^*)}.
\end{align*}
For any $(z^*, \ell, j_{\ell}, j^*)$, let
$Q \defeq X_{\ell} \cap X^{(z_{\ell})}_{j_{\ell}} \cap X^{(z^*)}_{j^*}$.
We can find $\sup_{x \in Q} \frac{c(x, z_{\ell})}{c(x, z^*)}$
using linear programming and binary search. Specifically,
\[ \sup_{x \in Q} \frac{c(x, z_{\ell})}{c(x, z^*)} > \alpha
\iff \{x \in Q: c(x, z_{\ell}) > \alpha c(x, z^*)\} \neq \emptyset. \]
Hence, checking if $\alpha$ upper-bounds this supremum is equivalent to
checking if a polyhedron is non-empty, which can be solved in polynomial time.

There are two special classes of problems where we don't need to use binary search
to compute $\sup_{x \in Q} \frac{c(x, z_{\ell})}{c(x, z^*)}$;
a single run of a linear program solver suffices:
\begin{enumerate}
\item \textbf{$c(x, z)$ is independent of $x$ for all $z \in Z$}:
    \\ Then the objective $\frac{c(x, z_{\ell})}{c(x, z^*)}$ becomes constant.
    Bin packing is one such problem because the cost is the number of bins,
    which depends on how the items are partitioned across bins, not on the size of the items.
\item \textbf{The input is scale-invariant}:
    \\ Formally, if $c(x, z^*)$ and $c(x, z_{\ell})$ are linear (not just affine) in $x$
    (i.e., $\exists a^*, a_{\ell} \in (\mathbb{R} \cup \{\infty\})^n$ such that
        $c(x, z^*) = (a^*)^Tx$ and $c(x, z_{\ell}) = a_{\ell}^Tx$ for all $x \in Q$),
    and $Q$ is a cone%
\footnote{A set $S \subseteq \mathbb{R}^n$ is a cone iff $\forall x \in S$ and
    $\forall \beta \in \mathbb{R}_{\ge 0}$, we have $\beta x \in S$.}, then
    \[ \sup_{x \in Q} \frac{c(x, z_{\ell})}{c(x, z^*)}
    = \sup_{x \in Q:\, c(x, z^*) \le 1} c(x, z_{\ell}). \]
    Makespan minimization is one such problem.
\end{enumerate}

\section{Computing Decision Trees}
\label{sec:cdt}

In this section, we describe how to obtain the decision tree for
an algorithm whose input space is $\mathbb{R}^n$.
The algorithm is given in an appropriate programming language (e.g., Python).
To warm up, let us first discuss circuits and how to recover a circuit
from its code using a technique we call \emph{gray-box access}.

\subsection{Recovering Circuits}
\label{sec:cdt:circuits}

A circuit is a recursively-defined structure.
A circuit is either a variable, or a constant,
or an operator with a sub-circuit for each of its operands.
A circuit can be represented as a directed acyclic graph (DAG).
It has a single source vertex, which represents the output of the circuit.
Variables and constants are sink vertices of the circuit.
See \cref{fig:circuits} for examples.

\begin{figure}[htb]
\centering
\begin{subfigure}{0.25\textwidth}
\centering
\begin{tikzpicture}
\node[cVar] (x) at (0,0) {$x$};
\node[cVar] (y) at (2,0) {$y$};
\node[cOp] (mult1) at (0,2) {$\times$};
\node[cOp] (mult2) at (2,2) {$\times$};
\node[cOp,cOut] (sub) at (1,3.5) {$-$};
\draw[cEdge] (sub) -- (mult1);
\draw[cEdge] (sub) -- (mult2);
\draw[cEdge] (mult1) to[bend left=30] (x);
\draw[cEdge] (mult1) to[bend right=30] (x);
\draw[cEdge] (mult2) to[bend left=30] (y);
\draw[cEdge] (mult2) to[bend right=30] (y);
\end{tikzpicture}

\caption{$x^2-y^2$}
\label{fig:circuits:quad1}
\end{subfigure}
\begin{subfigure}{0.25\textwidth}
\centering
\begin{tikzpicture}
\node[cVar] (x) at (0,0) {$x$};
\node[cVar] (y) at (2,0) {$y$};
\node[cOp] (add) at (0,2) {$+$};
\node[cOp] (sub) at (2,2) {$-$};
\node[cOp,cOut] (mult) at (1,3.5) {$\times$};
\draw[cEdge] (add) -- (x);
\draw[cEdge] (add) -- (y);
\draw[cEdge] (sub) -- (x);
\draw[cEdge] (sub) -- (y);
\draw[cEdge] (mult) -- (add);
\draw[cEdge] (mult) -- (sub);
\end{tikzpicture}

\caption{$(x+y)(x-y)$}
\label{fig:circuits:quad2}
\end{subfigure}
\begin{subfigure}{0.45\textwidth}
\centering
\begin{tikzpicture}
\node[cConst] (2) at (0,0) {$2$};
\node[cVar] (x) at (1,0) {$x$};
\node[cOp] (2x) at (0.5,1.3) {$\times$};
\draw[cEdge] (2x) -- (2);
\draw[cEdge] (2x) -- (x);
\node[cConst] (3) at (2,0) {$3$};
\node[cVar] (y) at (3,0) {$y$};
\node[cOp] (3y) at (2.5,1.3) {$\times$};
\draw[cEdge] (3y) -- (3);
\draw[cEdge] (3y) -- (y);
\node[cOp] (xy) at (1.5,2.4) {$+$};
\draw[cEdge] (xy) -- (2x);
\draw[cEdge] (xy) -- (3y);
\node[cVar] (z) at (3.5,2.1) {$z$};
\node[cOp] (lin) at (2.5,3.2) {$+$};
\draw[cEdge] (lin) -- (xy);
\draw[cEdge] (lin) -- (z);
\node[cConst] (10) at (4.5,2.9) {$10$};
\node[cOp,cOut] (out) at (3.5,4) {$\ge$};
\draw[cEdge] (out) -- (lin);
\draw[cEdge] (out) -- (10);
\end{tikzpicture}

\caption{$2x+3y+z \ge 10$}
\label{fig:circuits:lin}
\end{subfigure}
\caption{Examples of circuits represented as DAGs.
Variable nodes are rectangles. Output nodes have a double border.}
\label{fig:circuits}
\end{figure}
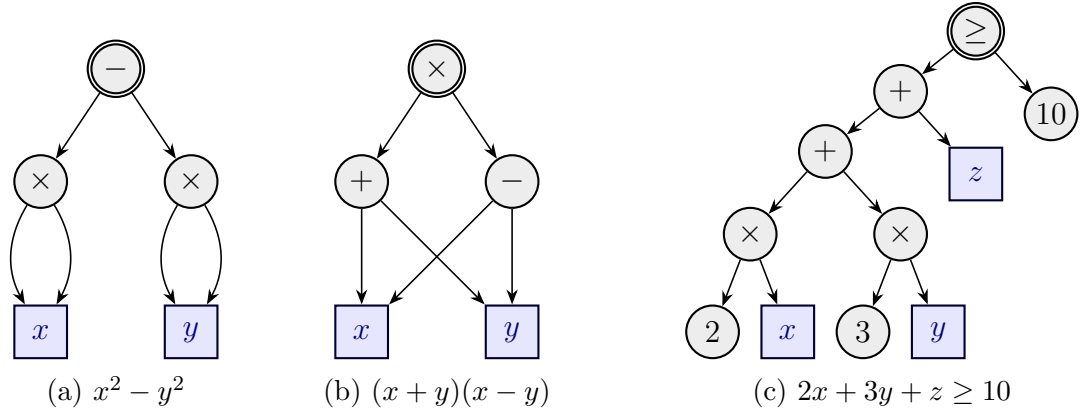

The functions \texttt{f1}, \texttt{f2}, and \texttt{f3} below in Python correspond to
circuits in \cref{fig:circuits:quad1,fig:circuits:quad2,fig:circuits:lin}, respectively.

\noindent\begin{minipage}[t]{0.45\textwidth}
\begin{lstlisting}[language=Python]
def f1(x, y):
    return x * x - y * y
\end{lstlisting}
\end{minipage}
\hfill
\begin{minipage}[t]{0.45\textwidth}
\begin{lstlisting}[language=Python]
def f2(x, y):
    s = x + y
    d = x - y
    return s * d
\end{lstlisting}
\end{minipage}
\begin{lstlisting}[language=Python]
def dot(x, y):
    # computes the dot product of two vectors
    return sum([xi * yi for (xi, yi) in zip(x, y)])

def f3(x, y, z):
    return dot([2, 3], [x, y]) + z >= 10
\end{lstlisting}

As evident from the above examples, especially \texttt{f3}, computing the circuit
for a function given as source code can be quite difficult.
This is because the function can execute arbitrary code,
and can even call other functions, including itself
and those in the programming language's standard library.
This task is about as difficult as writing a compiler/interpreter for the programming language.

Hence, we use a technique called \emph{gray-box access}, where we leverage
the programming language's compiler/interpreter to get the circuit for a function.
This technique only works with programming languages with certain special features.

The first key idea is to exploit duck typing or generics, i.e.,
we need not specify in advance the precise input types for functions.
E.g., for \texttt{f1}, even though we initially intended $x$ and $y$ to be real numbers,
\texttt{f1} also works with other input types, like complex numbers or matrices.
The second key idea is to implement circuits as a custom type
(called \emph{classes} in many programming languages).
Given circuits $C_1$ and $C_2$ and any binary operation $\circ$,
applying $\circ$ to $C_1$ and $C_2$ should return the circuit $C_1 \circ C_2$.
We combine these ideas, and pass circuits as inputs to our function,
and the output will be the circuit we want.

Here is a small self-contained runnable example in Python illustrating gray-box access.
This example outputs \texttt{((x + y) * z)}.

\begin{lstlisting}[language=python]
class Expr:
    def __add__(self, other):
        return BinExpr('+', self, other)

    def __mul__(self, other):
        return BinExpr('*', self, other)


class Var(Expr):
    def __init__(self, name):
        self.name = name

    def __str__(self):
        return str(self.name)


class BinExpr(Expr):
    def __init__(self, op, larg, rarg):
        self.op = op
        self.larg = larg
        self.rarg = rarg

    def __str__(self):
        return '({} {} {})'.format(str(self.larg),
            str(self.op), str(self.rarg))


def g(x, y, z):
    s = x + y
    return s * z


print(g(Var('x'), Var('y'), Var('z')))
\end{lstlisting}

A more comprehensive implementation of \texttt{Expr}, \texttt{Var}, and \texttt{BinExpr} can be found in
\href{https://github.com/sharmaeklavya2/code2dtree/blob/main/code2dtree/expr.py}{%
code2dtree/expr.py} in code2dtree \cite{code2dtree}.

\subsection{Recovering Decision Trees}
\label{sec:cdt:cdt}

The technique in \cref{sec:cdt:circuits} fails for functions with conditional
execution constructs, like \texttt{if} and \texttt{while} statements.
Indeed, such functions cannot be naturally represented as circuits,
as shown by the following function \texttt{signName}.

\begin{lstlisting}[language=Python]
def signName(x):
    if x > 0:
        return 'positive'
    elif x < 0:
        return 'negative'
    else:
        return 'zero'
\end{lstlisting}

However, we can represent them as decision trees, where each node is labeled with a circuit,
and internal nodes' circuits have boolean outputs.

We use the same high-level idea as before: run the function with \texttt{Var} objects as inputs.
However, the conditions in \texttt{if} and \texttt{while} statements
evaluate to \texttt{Expr} objects, not booleans.
Python tries to convert them to booleans, by calling the \verb|Expr.__bool__| method on them.
By running the function multiple times and strategically returning \texttt{True} or \texttt{False}
in \verb|Expr.__bool__|, we can systematically explore all branches of the function
and construct the decision tree.

We have implemented this idea in Python and made it available as a library called
\texttt{code2dtree} \cite{code2dtree}.
The following example shows how to use \texttt{code2dtree} to get
the decision tree for the sorting algorithm in Python's standard library
(the \href{https://docs.python.org/3/library/functions.html\#sorted}{sorted} function)
when given an input of length 3.

\noindent\textbf{Program:}
\begin{lstlisting}[language=python]
from code2dtree import func2dtree, getVarList
from code2dtree.nodeIO import printTree

x = getVarList('x', 3, 'simple')
dtree = func2dtree(sorted, [x])
printTree(dtree)
\end{lstlisting}

\noindent\textbf{Output:}
\begin{lstlisting}
if x1 < x0:
  if x2 < x1:
    return [x2, x1, x0]
  else:
    if x1 < x2:
      if x2 < x0:
        return [x1, x2, x0]
      else:
        return [x1, x0, x2]
    else:
      return [x1, x2, x0]
else:
  if x2 < x1:
    if x0 < x1:
      if x2 < x0:
        return [x2, x0, x1]
      else:
        return [x0, x2, x1]
    else:
      return [x2, x0, x1]
  else:
    return [x0, x1, x2]
\end{lstlisting}

\section{Implementation for LPT}
\label{sec:implement-lpt}

We implemented our method in Python for the LPT algorithm for makespan minimization.
The script at
\url{https://gist.github.com/sharmaeklavya2/4e4c52565b2d7aeceb074e870e6ad4aa}
takes $n$ (number of jobs) and $m$ (number of machines) as input
and outputs the best approximation factor and the corresponding hard example.

The script uses a few tricks to reduce its running time:
\begin{enumerate}
\item One can show that in any hard example, the machine having the highest load
    contains the last job. (Otherwise we can remove the jobs appearing after
    the max-loaded machine's last job. This doesn't affect the makespan of LPT,
    but may reduce the optimal makespan.)
    Restricting the search to such hard examples makes the cost function for LPT linear
    (instead of piecewise-linear). This shrinks the search space.

    To enforce this structure on hard examples,
    we can add additional constraints to each leaf node in the decision tree.
    Specifically, if the last job was assigned to machine $i$ in solution $z$,
    we add the constraints $\ell_i(x, z) \ge \ell_j(x, z)$ for all $j \in [n] \setminus \{i\}$.
    In our script, we do this indirectly: in LPT, we return not only an assignment,
    but also the index of the most loaded machine.
    Then we simply prune out leaves where the last job was not assigned to the most loaded machine.
    (C.f. function \texttt{check1} in the script.)

\item In \cref{sec:che}, we said that for problems where the input is scale-invariant,
    we maximize $c(x, z_{\ell})$ subject to the constraint $c(x, z^*) \le 1$.
    Since $c(x, z^*) \defeq \max_{i=1}^n \ell_i(x, z^*)$, we can replace
    $c(x, z^*) \le 1$ by $n$ constraints: $\ell_i(x, z^*) \le 1$ for all $i \in [n]$.
    Hence, we don't need to partition the input space into $n$ pieces
    such that $c(x, z^*)$ is linear in each piece.

\item For each leaf $\ell$, define
    \[ \beta_{\ell} \defeq \max_{z^* \in Z} \sup_{x \in X_{\ell}: c(x, z^*) \le 1} c(x, z_{\ell}). \]
    By relaxing the constraint $c(x, z^*) \le 1$ to the constraints
    $\|x\|_1 \le m$ and $x_1 \le 1$, we can upper-bound $\beta_{\ell}$ by
    \[ \gamma_{\ell} \defeq \sup_{x \in X_{\ell}: \|x\|_1 \le m, x_1 \le 1} c(x, z_{\ell}). \]
    If $\gamma_{\ell} = -\infty$ (i.e., the LP defining it is infeasible),
    or if $\gamma_{\ell} \le \beta_{\ell'}$ for some other leaf $\ell'$,
    then we remove $\ell$ from consideration. (C.f. function \texttt{check2} in the script.)
\end{enumerate}

Despite these tricks, the script runs quite slowly.
The script could find the optimal solution for $m=4$ machines and $n=9$ jobs (seed 0)
in around 4 minutes on an Intel i5 macbook.
The script runs on a single core, although our approach is highly parallelizable.

The approximation ratio of LPT is $(4m-1)/3m$, and the corresponding
hard example uses $n = 2m+1$ jobs \cite{graham1969bounds}.
Our program successfully recovers these examples for $m \le 4$.

\section{Conclusion}
\label{sec:conclusion}

For optimization problems whose input space is a subset of $\mathbb{R}^n$ (where $n$ is constant)
and whose output space is discrete, for any approximation algorithm,
we devise a technique to find its approximation ratio and the corresponding hard examples.
This technique assumes that the problem and the approximation algorithm are \emph{linear}.

We implement our decision tree recovery technique, and show how to use it to compute the
approximation ratio and hard examples for the LPT heuristic for makekspan minimization.

The biggest drawback of our technique is its slow speed in practice,
since the size of the output space and the decision tree grow exponentially in the input size.
But our technique holds promise even with these drawbacks.
For some problems, hard examples for even very small inputs can be insightful.
Moreover, if one can guess the optimal solution $z^*$ and/or the algorithm's output $z_{\ell}$
for the hard example, the hard example can be recovered more easily.

Perhaps our approach can be significantly enhanced if we can somehow guide
the search for the right $z^*$ and $\ell$, instead of a brute-force computation.
This can be done using, e.g., the branch-and-bound framework, or reinforcement learning.

Our decision tree recovery tool can be of independent interest,
e.g., to analyze the conditions under which an algorithm gives a certain output.

\paragraph{Acknowledgements.}
Eklavya Sharma was partially supported by NSF Grant CCF-2334461.

\appendix
\section{Optimizing over Polyhedra}
\label{sec:open-lp}

In this section, we show how to extend linear programming for non-closed polyhedra.
We do this by reducing it to the closed case.

Let $P \defeq \{x \in \mathbb{R}^n: (a_i^Tx = b_i \,\forall i \in E) \text{ and }
    (a_i^Tx \ge b_i \,\forall i \in I_1) \text{ and }
    (a_i^Tx > b_i \,\forall i \in I_2) \}$,
where $E$, $I_1$, and $I_2$ are indices over constraints. Let $I \defeq I_1 \cup I_2$.
For all $i \in E \cup I$, we have $b_i \in \mathbb{R}$ and $a_i \in \mathbb{R}^n$.
$E$ and $I_1$ may be empty.

$P$'s \emph{closure} is defined as $Q \defeq \{x \in \mathbb{R}^n:
    (a_i^Tx = b_i \,\forall i \in E) \text{ and }
    (a_i^Tx \ge b_i \,\forall i \in I) \}$.
$Q$'s \emph{direction cone} is defined as $D \defeq \{x \in \mathbb{R}^n:
    (a_i^Tx = 0 \,\forall i \in E) \text{ and }
    (a_i^Tx \ge 0 \,\forall i \in I) \}$.

Linear programming over $P$ is closely related to linear programming over $Q$.
Hence, we begin by stating a stronger form of \cref{thm:lin-prog} for $Q$.

\begin{lemma}[\cite{bazaraa2010linear}]
\label{thm:lin-prog-strong}
Given the polyhedron $Q$ and a vector $c \in \mathbb{R}^n$ as inputs,
there is an efficient algorithm that outputs a pair $(x^*, d)$, where
$x^* \in Q \cup \{\Null\}$ and $d \in D$ such that
\begin{tightenum}
\item (infeasible) If $x^* = \Null$, then $Q = \emptyset$.
\item (unbounded) If $x^* \neq \Null$ and $d \neq 0$, then $c^Td > 0$.
\item (optimal) If $x^* \neq \Null$ and $d = 0$, then $c^Tx^* \ge c^Ty$ for all $y \in Q$.
\end{tightenum}
\end{lemma}

Note that when $x^* \neq \Null$ and $d \neq 0$ in \cref{thm:lin-prog-strong},
then one can easily verify that $x^* + \alpha d \in Q$ for all
$\alpha \in \mathbb{R}_{\ge 0}$ using the definitions of $Q$ and $D$.

$P \subseteq Q$. Hence, if $Q = \emptyset$
(which we can check efficiently by \cref{thm:lin-prog-strong}),
then $P = \emptyset$, and we are done.
Now assume $Q \neq \emptyset$.

\subsection{Linear Feasibility}

Let us first start with the linear feasibility problem.
Here we need to check if $P \neq \emptyset$, and if so, output a point $x \in P$.

Define $R \defeq \{(z, x) \in \mathbb{R}_{\ge 0} \times \mathbb{R}^n:
    (a_i^Tx = b_i \,\forall i \in E) \text{ and }
    (a_i^Tx \ge b_i \,\forall i \in I_1) \text{ and }
    (a_i^Tx \ge b_i + z \,\forall i \in I_2)\}$.
For any constant $\delta \in \mathbb{R}_{\ge 0}$,
define $T(\delta) \defeq \{x: (\delta, x) \in R\}$.
Then $T(\delta)$ is a closed polyhedron for every $\delta \in \mathbb{R}_{\ge 0}$.
Note that $Q = T(0)$.

We can solve linear feasibility for $T(1)$ efficiently.
If $x \in T(1)$, then $x \in P$, since $T(1) \subseteq P$.
Now assume $T(1) = \emptyset$.

Consider the following (closed) linear programming problem:
$\displaystyle \max_{(z, x) \in R: z \le 1} z$.
\\ This problem has a feasible solution (of the form $(0, x)$) since $Q \neq \emptyset$.
Let $(\zhat, \xhat)$ be an optimal solution.
Then $P = \emptyset$ iff $\zhat = 0$.
Moreover, if $P \neq \emptyset$, then $\xhat \in P$.

\subsection{Linear Optimization}

Assume $P \neq \emptyset$. Let $\xhat$ be an arbitrary point in $P$.
Let $(x^*, d)$ be the solution to $\max_{x \in Q} c^Tx$,
as given by \cref{thm:lin-prog-strong}. Then $x^* \neq \Null$.
Using $\xhat$, $x^*$, and $d$, one can easily solve the linear programming problem for $P$.

\begin{lemma}
\label{thm:lin-prog-open}
If $d \neq 0$, then $\xhat + \alpha d \in P$ for all $\alpha \in \mathbb{R}_{\ge 0}$.
If $d = 0$, then for any $\eps \in \mathbb{R}_{> 0}$,
define $x_{\eps} \defeq (1-\eps)x^* + \eps\xhat$. Then $x_{\eps} \in P$
and $c^Tx_{\eps} \ge c^Tx^* - \eps(c^Tx^* - c^T\xhat)$.
\end{lemma}
\begin{proof}
The claims $\xhat + \alpha d \in P$ and $x_{\eps} \in P$
follow from the definitions of $P$, $Q$, and $D$.
\end{proof}

\begin{lemma}
\label{thm:lin-prog-open-sup}
If $P \neq \emptyset$, then
$\displaystyle \sup_{x \in P} c^Tx = \sup_{x \in Q} c^Tx$.
\end{lemma}
\begin{proof}
This is a simple corollary of \cref{thm:lin-prog-open}.
If $d \neq 0$, then both $P$ and $Q$ contain a ray of increasing objective value,
so $\sup_{x \in P} c^Tx = \sup_{x \in Q} c^Tx = \infty$.

If $d = 0$, then $\sup_{x \in Q} c^Tx = c^Tx^*$.
Since $P \subseteq Q$, we have $\sup_{x \in P} c^Tx \le \sup_{x \in Q} c^Tx = c^Tx^*$.
For all $\eps > 0$, we have
$c^Tx^* - \eps(c^Tx^* - c^T\xhat) \le c^Tx_{\eps} \le c^Tx^*$.
By picking a small enough $\eps$, we can get $c^Tx_{\eps}$
arbitrarily close to $c^Tx^*$.
\end{proof}

\end{document}